\documentclass[12pt,wide,reqno]{amsart}

\usepackage{graphicx}
\usepackage{amsmath}
\usepackage{amssymb}
\usepackage{amscd}
 \usepackage{tikz}
 \usepackage{mathrsfs}
\usepackage{pdfpages}
\usepackage{bbm}
\usepackage{amsthm}
\usepackage{amsfonts}
\usepackage{cite}

\usepackage{color}

\newtheorem{theorem}{Theorem}
\newtheorem{theorema}{Theorem}
\newtheorem{theoremb}{Theorem}
\newtheorem{theoremc}{Theorem}

\newtheorem{theoreme}{Theorem}
\newtheorem{rk}[theorema]{Remark}
\newtheorem{lem}[theoremb]{Lemma}
\newtheorem{prop}[theoremc]{Proposition}

\newtheorem{dfn}[theoreme]{Definition}

\newcommand{\comm}[1]{}

\renewcommand\a{\alpha}
\newcommand\af{\mathfrak{a}}
\renewcommand\b{\beta}

\newcommand\C{{\mathbb C}}

\newcommand\E{{\mathcal E}}

\newcommand\g{\mathfrak{g}}

\newcommand\h{\mathfrak{h}}

\renewcommand\l{\lambda}

\newcommand\op[1]{\mathop{\rm #1}\nolimits}

\newcommand\p{\partial}

\newcommand\Pp{{\let\mathcal\mathscr\mathcal P}}
\newcommand\Qq{{\let\mathcal\mathscr\mathcal Q}}
\newcommand\R{{\mathbb R}}

\newcommand\s{\mathfrak{s}}
\newcommand\si{\mathfrak{s}_\infty}

\begin{document}

 \title[Integrable PDE in 4D via symmetry]{Integrable dispersionless PDE in 4D, \\ 
 their symmetry pseudogroups \\ and deformations}
% Integrable symmetric deformations of heavenly type equations in 4D
 \author{Boris Kruglikov, Oleg Morozov}
 \address{Department of Mathematics and Statistics, NT-faculty, University of Troms\o, Troms\o\ 90-37, Norway}
 \email{boris.kruglikov@uit.no}
 \address{Department of Applied Mathematics, AGH University of Science and Technology, Krakow  30-059, Poland}
 \email{morozov@agh.edu.pl}
 \maketitle

 \begin{abstract}
We study integrable non-degenerate Monge-Amp\`ere equations of Hirota type in 4D and demonstrate that 
their symmetry algebras have a distinguished graded structure, uniquely determining the equations. 
This is used to deform these heavenly type equations into new integrable PDE of the second order with
large symmetry pseudogroups. 
We classify the obtained symmetric deformations and discuss self-dual hyper-Hermitian geometry of their solutions,
which encode integrability via the twistor theory.
 \end{abstract}

% Keywords: Monge-Ampere equation, symmetry pseudogroup, grading, Lax pair, area preserving transformations,
% recursion operator, self-duality, Petrov type.

\bigskip

 \section{Introduction and main results}

Solitonic integrable equations in 2D (i.e.\ with 2 independent variables) have few local symmetries, while
as a rule they have infinity of higher (generalized) symmetries.
In contrast, many dispersionless integrable PDEs in 3D and 4D have large local
symmetry pseudogroups. In this paper we discuss an important class of such dispersionless
differential equations in 4D and study their symmetries.

We consider integrable non-degenerate Monge-Amp\`ere equations of Hirota type in 4D.
These were classified by Doubrov and Ferapontov in \cite{DF}. There are 5 such non-linear equations
up to the natural (for this class) equivalence group $\op{Sp}(8)$.
These are heavenly-type equations, important in general relativity, and
it turns out that all of them have 4 copies of the infinite group
$\op{SDiff}(2)$ (of area preserving diffeomorphisms) in the symmetry pseudogroup,
and the equations can be distinguished by the graded structure of the corresponding Lie algebras.

The combinatorial structure behind the grading can be encoded by types similar
to Petrov types in classification of conformal Weyl tensors (or to classes of pencils of conics).
This rough invariant of the contact symmetry algebra can be also read off the geometry of the equation.

We reduce the symmetry algebra by omitting some graded component, i.e.
by restricting to a graded (infinite) Lie subalgebra, and we search integrable equations
within the class of symmetric equations. As the minimal differential invariants are of order 2
(this reflects the Lie remarkable property) the integrable families of PDE
we obtain can be called the symmetric integrable deformations.

The obtained integrable equations in 4D are new, and come in families.
Here are four PDE that deform the first and the modified heavenly equations,
Hussain's and the general heavenly equations:
 \begin{gather*}
u_{xt}u_{yz}-u_{xz}u_{yt}=q_z u_t-q_t u_z+b,\\
u_{yt}-u_{xt}u_{zz}+u_{xz}u_{zt}=Q\, u_{zt},\\
u_{xy}u_{zt}-u_{xz}u_{yt}=Q\, u_{xt},\\
u_{xy}u_{zt}-u_{xz}u_{yt}=Q\,(u_{xy}u_{zt}-u_{xt}u_{yz}).
 \end{gather*}
where $q=q(z,t)$, $b=b(z,t)$ and $Q=Q(t,u_t)$ are arbitrary functions. These equations are integrable
via Lax pair with spectral parameter and they also possess recursion operators. Moreover every
solution of any of this equation determines canonically a self-dual metric, so they can be obtained as reductions
of the (conformal) self-duality equation. In addition, these metric structures have hyper-Hermitian property.

Our method can be used to uncover many more integrable equations by deforming %, using the proposed approach,
from the known integrable models. % (we have a confirmation to this for other equations in 3D and 4D).

The structure of the paper is as follows. In Section \ref{S2} we describe the symmetry of the model equations,
exhibiting 4 copies of the Poisson algebra in the symmetry algebra, discuss natural grading and
compute the differential invariants. Then in Section \ref{S3} we compute the symmetric deformation, and
classify integrable equations among them. The main list is given in Theorem \ref{Thm3}.
In Proposition \ref{Thm3+} we classify some other branches, and then discuss the moduli of the obtained
symmetric families with respect to the natural equivalence. Finally in Section \ref{S4} we discuss the self-dual
geometry standing behind our classification, explain Petrov types via the geometry of the equations,
and compute the recursion operators for the deformed equations. In the Appendix we exhibit another method
of extending the integrable equations and relate it to the method of integrable symmetric deformations.

 \section{Symmetry of integrable 4D Monge-Amp\`ere equations}\label{S2}

Let $M^4$ be a 4-dimensional manifold with coordinates $x^i$, $i=1,\dots,4$.

The integrable Monge-Amp\`ere equations of Hirota type on $M^4$ have the form
($u_{ij}$ are the second derivatives of $u=u(x^i)$, $1\le i\le j\le4$)
 $$
F(u_{11},u_{12},\dots,u_{44})=0.
 $$
The group $\op{Sp}(8)$ preserves the class of these equations (this is the natural equivalence group
for this class) and every integrable PDE of such form is $\op{Sp}(8)$ equivalent (over $\C$) to
one of the following equations:

 \begin{enumerate}
\item[(O)] $u_{11}-u_{22}-u_{33}-u_{44}=0$ \ (linear wave equation) \vphantom{$\frac{1^1}{1^1}$}
\item[(N)] $u_{24}-u_{13}+u_{11}u_{22}-u_{12}^2=0$ \ (second heavenly equation \cite{Pb}) \vphantom{$\frac{A^A}{A^A}$}
\item[(D)] $u_{14}u_{23}-u_{13}u_{24}=1$ \ (first heavenly equation \cite{Pb}) \vphantom{$\frac{A^A}{A^A}$}
\item[(III)] % $u_{13} = u_{12}u_{44}-u_{14}u_{24}$
           $u_{24} = u_{14}u_{33}-u_{13}u_{34}$ \ (modified heavenly equation) \vphantom{$\frac{A^A}{A^A}$}
\item[(II)]  % $u_{11}+u_{22}+u_{13}u_{24}-u_{14}u_{23}=0$
           $u_{14}=u_{13}u_{24}-u_{12}u_{34}$ \ (Hussain's equation\footnote{The usual Hussain's equation
           $u_{11}+u_{22}+u_{13}u_{24}-u_{14}u_{23}=0$
           is equivalent to this form via a complex change of coordinates.}) \vphantom{$\frac{A^A}{A^A}$}
\item[(I)] $\alpha u_{12}u_{34} + \beta u_{13}u_{24}-(\alpha+\beta) u_{14}u_{23} = 0$ \ (general
heavenly). \vphantom{$\frac{A^A}{A^A}$}
 \end{enumerate}

The list and classification are due to \cite{DF}, but the enumeration of the equations is ours
and it will be explained below. Let us first describe the contact symmetry algebras of these PDEs.

 \subsection{Symmetry algebras of the equations}\label{SAE}

Below we use the following conventions: any contact symmetry $X_f$ is uniquely represented by its
generating function $f\in C^\infty(J^1(\R^4))$;
the contact symmetry algebra $\g=\mathfrak{cont}(\E)$ of an equation $\E$
is a Lie algebra with respect to the natural Jacobi bracket on $C^\infty(J^1(\R^4))$ given by $[X_f,X_g]=X_{\{f,g\}}$;
as an abstract Lie algebra $\g$ is obtained from its infinite part (subalgebra) $\si$ by
a finite-dimensional "right" extension $\s_\diamond$ (extension by derivations):
 $$
0\to\si\longrightarrow\g\longrightarrow\s_\diamond\to0.
 $$
The index indicates the grading and the superscript indicates the commuting branches, so that
$[\g_i^\a,\g_j^\b]\subset\delta^{\a\b}\g_{i+j}^\a$ (if $i+j$ is not an admissible grading, then $\g^\a_{i+j}=0$).
For the infinite part the grading will be denoted by $\af_i$, so $\si=\oplus\af_i^\gamma$
(superscript $\gamma$ can be
empty or prime); the grading of the finite-dimensional part is $\h_i$: $\s_\diamond=\oplus\h_i$.

Everywhere in this section $\af_i\simeq\Pp=C^\infty(M)$ for some choice $M=\R^2(p,q)$;
notice the isomorphism $\Pp/\R\simeq\op{SDiff}(2)$ (the latter is the algebra of Hamiltonian vector fields on $M$,
that is the Lie algebra of the group of area preserving transformations,
denoted, by the abuse of notations, by the same symbol),
$f\mapsto X_f$. We will indicate below the choice of $(p,q)$ and equip this $M$ with the standard Poisson bracket
 $$
\{A(p,q),B(p,q)\}=A_pB_q-A_qB_p.
 $$
Then the bracket $[\af_i^\a,\af_j^\a]\subset\af_{i+j}^\a$ is given by the rule
(again, if $i+j$ fall out of range of the grading the corresponding bracket has to be zero):
 $$
[A_i,B_j]=\{A,B\}_{i+j}
 $$

 \begin{theorem}\label{Thm1}
For nonlinear Monge-Amp\`ere equations from the above list $\si$ is the derived algebra of $\g$
(more precisely $\si=[\g,\g]$ in the cases {\rm D}, {\rm II}, {\rm I} and
$\si=[[\g,\g],[\g,\g]]$ in the cases {\rm N}, {\rm III}).
We have:
 \begin{itemize}
\item[(N)] $\si=\af_0\oplus\af_1\oplus\af_2\oplus\af_3$; \hskip23pt $\s_\diamond=\R^2\ltimes\R^1$.
 \vphantom{$\frac{A^A}{A^A}$}
\item[(D)] $\si=(\af_0'\oplus\af_1')\oplus(\af_0''\oplus\af_1'')$; \hskip4pt $\s_\diamond=\R^3$.
 \vphantom{$\frac{A^A}{A^A}$}
\item[(III)] $\si=(\af_0'\oplus\af_1'\oplus\af_2')\oplus\af_0''$; \hskip13pt $\s_\diamond=\R^2\ltimes\R^1$.
 \vphantom{$\frac{A^A}{A^A}$}
\item[(II)] $\si=(\af_0'\oplus\af_1')\oplus\af_0''\oplus\af_0'''$; \hskip11pt $\s_\diamond=\R^1$.
 \vphantom{$\frac{A^A}{A^A}$}
\item[(I)] $\si=\af_0'\oplus\af_0''\oplus\af_0'''\oplus\af_0^{\textsl{iv}}$; \hskip18pt $\s_\diamond=\R^1$.
 \vphantom{$\frac{A^A}{A^A}$}
 \end{itemize}
 \end{theorem}

 \begin{proof}
The proof of this theorem is the direct computation. Non-routine part is to organize the
result to observe the graded structure. Denoting\footnote{This order of variables $(x,y,z,t)=(x^1,x^2,x^3,x^4)$
is fixed throughout the paper.}
$x=x^1,y=x^2,z=x^3,t=x^4$, we get the following output:

\smallskip

(N) Here $\oplus_{i=0}^3\af_i\simeq\Pp_0\oplus\Pp_1\oplus\Pp_2\oplus\Pp_3$, $M=\R^2(t,z)$ with
the above Poisson bracket, and the (inverse) isomorphism is given by the rule \cite{KM}:
$A_0\mapsto A_{0t}u_z-A_{0z}u_t+(A_{0tt}x+A_{0tz}y)u_y-(A_{0tz}x+A_{0zz}y)u_x-\frac16\nabla^3A_0\in\af_0$,
$A_1\mapsto A_{1t}u_y-A_{1z}u_x-\frac12\nabla^2A_1\in\af_1$,
$A_2\mapsto \nabla A_2\in\af_2$, $A_3\mapsto A_3\in\af_3$, where $\nabla=x\p_t+y\p_z$, $A_i=A_i(t,z)\in\Pp_i$.

The finite-dimensional part is $\h_0=\langle f_0',f_0''\rangle$, $\h_1=\langle f_1\rangle$, where
$f_0'=2u-xu_x-yu_y-zu_z-tu_t$, $f_0''=3u-xu_x-yu_y$, $f_1=tu_x+zu_y$
($\h_0\oplus\h_1$ is a 3-dimensional nilpotent Lie algebra).

\smallskip

(D) In this case $\af'_0\oplus\af'_1\simeq\Pp'_0\oplus\Pp'_1$, $M'=\R^2(x,y)$;
$\af''_0\oplus\af''_1\simeq\Pp''_0\oplus\Pp''_1$, $M''=\R^2(t,z)$.
The (inverse) isomorphism was computed in \cite{KM}:
$A_0'\mapsto A'_{0x}u_y-A'_{0y}u_x\in\af'_0$, $A'_1\mapsto A'_1\in\af'_1$,
$A_0''\mapsto A''_{0t}u_z-A'_{0z}u_t\in\af''_0$, $A''_1\mapsto A''_1\in\af'_1$,
where $A'_i=A'_i(x,y)\in\Pp'_i$, $A''_i=A''_i(t,z)\in\Pp''_i$.

The finite part is $\h_0=\langle f_0',f_0'',f_0'''\rangle$, where
$f_0'=xu_x-yu_y$, $f_0''=tu_t-zu_z$, $f_0'''=xu_x+yu_y-zu_z-tu_t$
($\h_0$ is an Abelian algebra).

\smallskip

(III) Here we have: $\oplus_{i=0}^2\af'_i\simeq\Pp'_0\oplus\Pp'_1\oplus\Pp'_2$, $M'=\R^2(x,y)$;
$\af''_0\simeq\Pp''_0$, $M''=\R^2(t,u_t)$. The (inverse) isomorphism is the following:
$A'_0\mapsto A'_{0y}u_x-A'_{0x}u_y+\frac12(A_{0yy}z^2-2A'_{0xy}zu_z+A'_{0xx}u_z^2)\in\af'_0$,
$A'_1\mapsto A'_{1y}z-A'_{1x}u_z\in\af'_1$, $A'_2\mapsto A'_2\in\af'_2$;
$A''_0\mapsto A''_0\in\af''_0$, where $A'_i=A'_i(x,y)\in\Pp'_i$, $A''_0=A''_0(t,u_t)\in\Pp''_0$.

The finite-dimensional part is $\h_0=\langle f_0',f_0'',f_0'''\rangle$, where
$f_0'=2u-zu_z$, $f_0''=u+yu_y$, $f_0'''=u_y$ ($\h_0$ is a 3-dimensional nilpotent Lie algebra).

\smallskip

(II) In this case $\af'_0\oplus\af'_1\simeq\Pp'_0\oplus\Pp'_1$, $M'=\R^2(y,z)$;
$\af''_0\simeq\Pp''_0$, $M''=\R^2(x,u_x)$; $\af'''_0\simeq\Pp'''_0$, $M'''=\R^2(t,u_t)$.
The (inverse) isomorphism is the following:
$A'_0\mapsto A'_{0z}u_y-A'_{0y}u_z\in\af'_0$, $A'_1\mapsto A'_1\in\af'_1$;
$A''_0\mapsto A''_0\in\af''_0$; $A'''_0\mapsto A'''_0\in\af'''_0$,
where $A'_i=A'_i(y,z)\in\Pp'_i$, $A''_0=A''_0(x,u_x)\in\Pp''_0$, $A'''_0=A'''_0(t,u_t)\in\Pp'''_0$.

The finite-dimensional part is $\h_0=\langle u-yu_y\rangle\simeq\R^1$.

\smallskip

(I) Here $\af'_0\simeq\Pp'_0$, $M'=\R^2(x,u_x)$;
$\af''_0\simeq\Pp''_0$, $M''=\R^2(y,u_y)$; $\af'''_0\simeq\Pp'''_0$, $M'''=\R^2(z,u_z)$;
$\af^{\textsl{iv}}_0\simeq\Pp^{\textsl{iv}}_0$, $M^{\textsl{iv}}=\R^2(t,u_t)$.
The (inverse) isomorphism is:
$\Pp^\a_0\ni A^\a_0\mapsto A^\a\in\af^\a_0$ ($\a$ is any number of primes).

The finite-dimensional part is $\h_0=\langle u\rangle\simeq\R^1$.
 \end{proof}

 \begin{rk}
For the linear case {\rm(O)} the infinite part $\si$ of the symmetry algebra is the Abelian space modeled on
the solution space $\op{Sol}({\rm O})$ of the equation (symmetries are shifts by solutions),
and the finite-dimensional part $\s_\diamond$
is the full conformal group $so(2,4)$\footnote{This depends on the signature; we indicate the
real form of the Lie algebra $D_3$ corresponding to the Lorenzian choice of the wave operator in (O).}
extended by the algebra $\R$ of scalings of $u$. The derived series of $\s_\diamond$ stabilizes on
$so(2,4)$. Thus this case differs from the 5 nonlinear integrable equations: $\si$ is not obtained as
the derived algebra, but on the contrast $so(2,4)$ plays the role of the Levi factor in the decomposition
$\g=\si\rtimes(so(2,4)\ltimes\R)$.
In addition the functional dimension of $\si$ is bigger than that of the other 5 cases:
2 functions of 3 variables vs 4 functions of 2 variables.
 \end{rk}

 \subsection{Lie remarkable property}\label{LRP}

Let us consider the contact Lie algebra of vector fields and prolong it to the higher jets.
Recall that a differential invariant of order $k$ is such a function $I$ on $J^k$ that it is constant
along the flow of the fields from the prolonged algebra: $L_{\hat X}I=0$ $\forall X\in\g$
($\hat X$ denotes the prolongation of the field $X$).

We have started from the list of 2nd order PDE and constructed their symmetry algebras. Now we
go in the opposite direction and compute the lowest order differential invariants of these algebras
(we'll even get the stronger statement by restricting to $\si\subset\g$).
The following claim is obtained by a straightforward computation.

 \begin{theorem}\label{Thm2}
There are no 1st order differential invariants for $\si$.

The only 2nd order differential invariant $I$ of the algebra $\si$,
up to the natural gauge $I\mapsto F(I)$, is the following:
 \begin{itemize}
\item[(O)] $I_{\rm O}=u_{11}-u_{22}-u_{33}-u_{44}$ (the functions $x^i$ are also invariant)
                                            %, but they are $\s_\diamond$-noninvariant;
 \vphantom{$\frac{A^A}{A^A}$}
\item[(N)] $I_\text{\rm N}=u_{24}-u_{13}+u_{11}u_{22}-u_{12}^2$;
 \vphantom{$\frac{A^A}{A^A}$}
\item[(D)] $I_{\rm D}=u_{14}u_{23}-u_{13}u_{24}$;
 \vphantom{$\frac{A^A}{A^A}$}
\item[(III)] $I_{\rm III}=(u_{24}-u_{14}u_{33}+u_{13}u_{34})/u_{34}$;
 \vphantom{$\frac{A^A}{A^A}$}
\item[(II)] $I_\text{\rm II}=(u_{13}u_{24}-u_{12}u_{34})/u_{14}$;
 \vphantom{$\frac{A^A}{A^A}$}
\item[(I)] $I_{\rm I}=(u_{12}u_{34}-u_{13}u_{24})/(u_{12}u_{34}-u_{14}u_{23})$.
 \vphantom{$\frac{A^A}{A^A}$}
 \end{itemize}
 \end{theorem}
Of course, the equations $I=c$ for the indicated functions $I$ are invariant under the bigger symmetry group
$\g=\si\rtimes\s_\diamond$, and in the case O the additional zero-order invariants $x^i$
disappear for $\g$, leaving only the indicated (now relative) differential invariant $I$.

The constant $c$ in the invariant equations $I=c$ can be set to 0 in the cases O, N and III
(by a contact transformation preserving the equation class), and to 0 or 1
in the cases D and II (1 is the regular value\footnote{For the value zero any of the two equations
is reduced to a 1st order PDE in 2D,
 % $u_{12}u_{34}=u_{13}u_{24}$ implies $u_2=F(u_1,x^1,x^2)$
so it is easily integrable and will be excluded from further consideration.});
$c$ cannot be changed (an essential parameter) in the case I.
Thus, up to equivalence, we recover the 5 integrable PDEs from their symmetry groups.

Therefore the considered PDEs are {\em Lie remarkable\/} equations, in accordance with terminology of \cite{MOV}.
We are going to use this property to obtain new integrable equations.

 \section{Symmetric deformation of Monge-Amp\`ere equations}\label{S3}

Now we are going to deform the equations preserving some degree of symmetry and integrability.

 \subsection{Symmetric deformation}\label{SD}

The idea is the following. Consider a Lie remarkable equation $F=0$, say of the second order.
Let $\s$ be the whole symmetry algebra and $\tilde\s\subset\s$ a subalgebra, such that the algebra of
differential invariants of $\tilde\s$ has more than one generator of order $\le2$. The new $\tilde\s$-invariant
equations $\tilde F=0$ (depending on functional parameters) will be considered as deformations of $F=0$.

In particular, for our case let $\si=\tilde\s\oplus\check{\s}$ be decomposition of the infinite part
of the symmetry into two commuting ideals.
Let us look for $\tilde\s$-invariant equations. % Of course, the original equation is one of them, but
We shall demonstrate that in all our cases the integrability constraint selects a subclass of
new equations, still integrable but not contact equivalent to the original $F=0$.

 \begin{dfn}
We call the family of symmetric equations obtained by the above approach the {\em symmetric deformation\/}
of the equation $F=0$.
 \end{dfn}

We will not touch upon the linear wave equation. Among the other 5 integrable Monge-Amp\`ere equations
the second Pleba\'{n}ski equation (N) does not have a splitting of $\si$ of the required type.
This means we cannot make {\em symmetric deformation\/} per se. And indeed, the mere knowing of the
zero-grading part $\Pp\simeq\af_0\subset\g$ recovers the second Pleba\'{n}ski equation uniquely \cite{KM}
(by a sequence of Lie algebra extensions). Moreover the nilpotent part $\af_1\oplus\af_2\oplus\af_3\subset\s$
of the symmetry algebra still has only one differential invariant (the 2nd Pleba\'{n}ski equation).
From this perspective we would call the equation of type N {\em symmetric deformation rigid\/}.

For the other 4 equations the idea works nicely. Below we indicate the integrable equations
obtained by the deformation, i.e. the integrable symmetric deformations (the class of
symmetric deformations is wider); the functional
parameters entering the formulae are arbitrary.

 \subsection{Classification of symmetric deformations}\label{SD1}

Let us indicate that by {\em integrable\/} equation we mean an equation possessing a
dispersionless Lax pair (the other signs of integrability are discussed in
the next section), i.e. the two vector fields $V,W$ on $M^4$ depending on
the jets of $u$ and the additional "spectral parameter" $\lambda$ such that
the rank 2 vector distribution $\langle V,W\rangle$ is Frobenius-integrable due to the equation.

 \begin{theorem}\label{Thm3}
The following are the integrable symmetric deformations of the above Monge-Amp\`ere equations,
corresponding to the largest ideal $\tilde\s\subset\s_\infty$ (in terms of the gradings) in the symmetry algebra $\s_\infty\subset\g$:
 \begin{alignat}{2}
 \rm{(D)} & \qquad &&
u_{14}u_{23}-u_{13}u_{24}=\p_zq(z,t)\,u_4-\p_tq(z,t)\, u_3+b(z,t),\vphantom{\frac{a}{a}}
\label{dfhe} \\
 \rm{(III)} & \qquad &&
u_{24}-u_{14}u_{33}+u_{13}u_{34}=Q(t,u_4)\, u_{34}, \vphantom{\frac{a}{a}}
\label{dmhe} \\
 \rm{(II)} & \qquad &&
u_{12}u_{34}-u_{13}u_{24}=Q(t,u_4)\, u_{14}, \vphantom{\frac{a}{a}}
\label{dhhe} \\
 \rm{(I)} & \qquad &&
u_{12}u_{34}-u_{13}u_{24}=Q(t,u_4)\,(u_{12}u_{34}-u_{14}u_{23}). \vphantom{\frac{a}{a}}
\label{dghe}
  \end{alignat}
 \end{theorem}

 \begin{proof}
(D) Let us reduce the infinite part of the symmetry algebra $\si$ to the ideal
$\tilde\s=\af'_0\oplus\af_1'$. The differential invariants of the second order of this algebra
is $F=F(I,z,t,u_3,u_4,u_{33},u_{34},u_{44})$, where $I=I_{\rm D}$.
In other words, the general form of the symmetric deformation in question is $I=Q(z,t,u_3,u_4,u_{33},u_{34},u_{44})$.

Now the integrability condition implies that $Q$ does not depend on $u_{ij}$
and it depends linearly on $u_i$ so:
 $$Q=\p_zq(z,t)\,u_4-\p_tq(z,t)\, u_3+b(z,t).$$
The Lax pair is ($q=q(z,t)$):
 \begin{equation}\label{ourD}
 \begin{array}{c}
\,V\,=Q\,\p_x+(\l+q)\,u_{14}\,\p_z-(\l+q)\,u_{13}\,\p_t,\vphantom{\frac{A^A}{A^A_A}}\\
W=Q\,\p_y+(\l+q)\,u_{24}\,\p_z-(\l+q)\,u_{23}\,\p_t.\vphantom{\frac{A^A_A}{A^A}}
 \end{array}
 \end{equation}

\smallskip

(III) The largest ideal in $\s_\infty$ is $\tilde\s=\af_0'\oplus\af_1'\oplus\af_2'$,
and the general differential invariant of order 2
of this algebra is $F=F(I,z,u_4)$, where $I=I_{\rm III}$.
Thus the invariant equation is $I=Q$, where $Q=Q(t,u_4)$.
The integrability conditions hold identically, and the
 %(nonlinear\footnote{Remark on producing a linear Lax pair from a nonlinear...} in $\l$)
Lax pair is
 \begin{gather*}
V=-u_{34}\,\p_x+u_{14}\,\p_z+(\l+Q)\,\p_t,\\
W\,=\,-u_{33}\,\p_x+\p_y+(\l+u_{13})\,\p_z.\,
 \end{gather*}

\smallskip

(II) The minimal way to reduce the symmetry algebra is clearly the ideal
$\tilde\s=\af_0'\oplus\af_1'\oplus\af_0''$.
The general second order differential invariant for it is $F=F(I,t,u_4,u_{44})$, where $I=I_{\rm II}$.
Again the integrability constraints $F$ to be independent of $u_{44}$, so
the invariant equation is $I=Q$, where
$Q=Q(t,u_4)$. Now the integrability conditions hold
identically, and the Lax pair for this PDE is
 \begin{gather*}
\,V\,=\l\,u_{24}\,\p_x-Q\,u_{14}\,\p_y-(\l-Q)\,u_{12}\,\p_t,\\
W=\l\,u_{34}\,\p_x-Q\,u_{14}\,\p_z-(\l-Q)\,u_{13}\,\p_t.
 \end{gather*}

\smallskip

(I) Here up to re-numeration of the coordinates we can choose
$\tilde\s=\af_0'\oplus\af_0''\oplus\af_0'''$, and the invariant equation with respect to this algebra is
$I=Q$, where $I=I_{\rm I}$ and $Q=Q(t,u_4)$.
The integrability condition holds identically and the Lax pair is
 \begin{equation}
\begin{array}{lcl}
V&=&(\l+1)\,u_{24}\,\p_x-Q\,u_{14}\,\p_y-(\l+1-Q)\,u_{12}\,\p_t,
\\
W&=&\l\,u_{24}\,\p_z-(Q-1)\,u_{34}\,\p_y-(\l+1-Q)\,u_{23}\,\p_t.
\end{array}
\label{dghe_Lax_pair}
 \end{equation}
This finishes the list of most symmetric integrable deformations.
 \end{proof}

 \begin{rk}
The Lax pairs presented in the proof do not contain derivation with respect to
the spectral parameter. Integrable equations with such Lax pairs are of special importance, as they
correspond to hyper-Hermitian 4D metrics on the solutions of the equation, see \cite{DunP,C}.
  \end{rk}

We do not investigate the most general ansatz for a subalgebra $\tilde\s\subset\s$, but
work in the paradigm described at the beginning of this section.
For instance, in the case N, if we choose $\tilde\s=\af_3$, the invariant function is
$F=F(x,y,z,t,u_1,u_2,u_{11},u_{12},u_{13},u_{14},u_{22},u_{23},u_{24}),$
and this is too general ansatz to search for integrable cases.
 % The situation with $\tilde\s=\af_2\oplus\af_3\subset\s$ is not much better.

Similarly for the case D the invariant equations for $\tilde\s=\af_1'$ are too big for
classification, but they do contain integrable branches, not covered above. For instance,
the dispersionless PDE
 $$u_{14}u_{23}-u_{13}u_{24}=Q, \quad Q=a(t)\,(c\, u_2+q(y))$$
is integrable with the following Lax pair (below $D_z,D_t$ denote the total derivatives):
 \begin{gather*}
V=-\l\,\p_x+u_{14}\,\p_z-u_{13}\,\p_t+\frac{u_{14}\,D_zQ-u_{13}\,D_tQ}Q\,\l\,\p_\l,\\
W=-\l\,\p_y+u_{24}\,\p_z-u_{23}\,\p_t+\frac{u_{24}\,D_zQ-u_{23}\,D_tQ}Q\,\l\,\p_\l.
 \end{gather*}
(notice that if the constant $c\neq0$ in the expression for $Q$,
then we can be normalize $c=1$, $q(y)=0$).

Thus we covered the possible symmetric deformations provided by our method in the cases I, D, N,
but the cases III and II have other branches (than those of Theorem \ref{Thm3}).

 \begin{prop}\label{Thm3+}
The following are the integrable symmetric deformations of the above Monge-Amp\`ere equations
in the cases {\rm III} and {\rm II} for the indicated choice of the symmetry ideal:
 %(modulo the natural equivalence):
       \vskip5pt
\noindent{\rm(III)} Choice $\tilde\s=\af_0''$. Then
$u_{24}-u_{14}u_{33}+u_{13}u_{34}=Q\cdot u_{34}$, where
 \begin{multline*}
Q=
\frac1{2\Delta}\Bigl((F_x-FF_y)u_3^2+2(G_x-GF_y-zF_y)u_3+H\Bigr)+\Delta\cdot L-\frac1\Delta(Fu_2+u_1)
 \end{multline*}
with $\Delta=Fu_3+G+z$, and $F,G,H,L$ are functions of $(x,y)$ satisfying
 $$
L_x+(FL)_y=\tfrac12F_{yy}.
 $$
The singular branches are given by the above formula for the equation with
$Q=\dfrac{u_2}{u_3}-k(x)\,u_3$, $Q=z\cdot\p_y f(x,y)-\p_x f(x,y)\cdot u_3$,
$Q=z\,f(y)$.

\smallskip

\noindent{\rm(II)} Choice $\tilde\s=\af_0''\oplus\af_0'''$.
Then $u_{12}u_{34}-u_{13}u_{24}=Q\cdot u_{14}$, where
$Q=\p_y f(y,z)\cdot u_3-\p_z f(y,z)\cdot u_2+h(y,z)$.
 \end{prop}

 \begin{proof}
(III) The indicated possibility for III $\tilde\s=\af_0''$
is the other choice in addition to that of Theorem \ref{Thm3}.
Then an invariant deformation
is $F=F(I,x,y,z,u_1,u_2,u_3)$ (the second order terms disappear as in the case D of Theorem \ref{Thm3}),
where $I=I_{\rm III}$.  In other words,
the invariant equation is again $I=Q$ with $Q=Q(x,y,z,u_1,u_2,u_3)$.

Here the integrability conditions do not hold identically, imposing a system of equations
on $Q$. The solution branches into the general case indicated above (with an ODE constraint)
and the following three cases (modulo the natural equivalence):
 $$
Q=\frac{u_2+z\cdot\p_yf(x,y)+h(x,y)}{u_3+f(x,y)}-\p_xf(x,y)-k(x)\,(u_3+f(x,y)),
 $$
and we can achieve $f=h=0$ by a point transformation (change of coordinates in $J^0=\R^5(x,y,z,t,u)$);
 $$
Q=z\cdot\p_y f(x,y)-(\p_x f(x,y)+k(x))\cdot u_3+h(x,y),
 $$
 and we can achieve $k=h=0$ in $I=Q$ by a point transformation;
 $$
Q=z\,f(y)+h(x,y),
 $$
and we can achieve $h=0$ in $I=Q$ by a point transformation.
These are precisely the subcases 2-4 in case III of the theorem.

The Lax pair is given for all subcases by the formula
 \begin{gather*}
V=-u_{34}\,\p_x+u_{14}\,\p_z+\l\,\p_t+(u_{24}Q_{u_2}+u_{34}Q_{u_3})\,\l\,\p_\l,\\
W=-u_{33}\,\p_x+\p_y+(\l+u_{13}-Q)\,\p_z+(Q_z+u_{23}Q_{u_2}+u_{33}Q_{u_3})\,\l\,\p_\l.
 \end{gather*}

\smallskip

(II) In case II we have two essentially different possibilities to
reduce the symmetry algebra: one considered in Theorem \ref{Thm3} and the other
$\tilde\s=\af_0''\oplus\af_0'''$ (the choice
$\tilde\s=\af_0'\oplus\af_1'\oplus\af_0'''$ is equivalent to the first one).
In this case we have to study the equation $I=Q$, where $Q=Q(y,z,u_2,u_3)$.
The integrability condition imposes the constraint, which resolves so:
$Q=\p_y f(y,z)\cdot u_3-\p_z f(y,z)\cdot u_2+h(y,z)$.
The Lax pair
 \begin{gather*}
\,V\,=\frac{\l}{\l-1}\,u_{34}\,\p_x-\frac{u_{14}}{\l-1}\,\p_z-u_{13}\,\p_t
+\p_z f(y,z)\cdot u_{14}\,\l\,\p_\l,\\
W=\frac{\l}{\l-1}\,u_{24}\,\p_x-\frac{u_{14}}{\l-1}\,\p_y-u_{12}\,\p_t
+\p_y f(y,z)\cdot u_{14}\,\l\,\p_\l.
 \end{gather*}
for the deformed PDE $I=Q$ yields the integrability.
 \end{proof}

%
% \begin{rk}
% Lax pairs of the integrable equations of this theorem contain derivations by the spectral
% parameter that cannot be removed. Indeed, the hyper-Hermitian property is equivalent to
% existence of a local (to the correspondence 5-space $M^5(x,y,z,t,\lambda)$) 1st integral,
% or equivalently a twistor function, which is not the case here.
% \end{rk}

  \subsection{Two remarks about normal forms of deformations}

Notice that the symmetry group of non-deformed PDEs, as described in Section \ref{SAE},
still acts on deformed PDEs and so can be considered as a natural equivalence group
(actually this is the quotient of the full symmetry group by its normal subgroup that is the symmetry
group of the deformed equation).
Consequently not all equations considered in Theorem \ref{Thm3} are different, and some are contact equivalent.

We already made such identification in Proposition \ref{Thm3+} for the last three sub-cases
of type III, by noticing (see the proof) that some functional parameters can be eliminated.
Doing so in all cases (i.e. reducing to normal forms all of which are non-equivalent) is
cumbersome, as the number of cases becomes too large.

Let us make two remarks illustrating the nature of this problem.

The first is that there can be a situation with no local invariants, yet with non-trivial global
invariants. Consider, for example, the second sub-case of type II (the one from Proposition \ref{Thm3+}):
 $$
u_{12}u_{34}-u_{13}u_{24}=(\p_y f(y,z)\cdot u_3-\p_z f(y,z)\cdot u_2+h(y,z))\, u_{14}.
 $$
A straightforward check shows
that change $u\mapsto u+g(y,z)$ results in the change $h\mapsto h+\{f,g\}$, where
 $$\{f,g\}=f_yg_z-f_zg_y$$
is the standard Poisson bracket on the plane $\R^2(y,z)$.
Thus locally we can eliminate the function $h=h(y,z)$ by a point transformation.
However globally there is an obstruction to doing so. For instance, if the manifold $M^4(x,y,z,t)$
fibers over a compact surface $\Sigma^2(y,z)$, then $\int h\,dy\wedge dz$ is an invariant
of the indicated gauge.

The second remark is that some known equivalence problems arise in classification.
Consider, for example, type I when the infinite part of the symmetry pseudo-group is the direct
sum of 4 copies of the group of area preserving transformations $\op{SDiff}(2)$,
see Theorem \ref{Thm1}. The group of area preserving transformations of $\R^2(t,u_4)$
preserves the class of deformations of type I from Theorem \ref{Thm3}, but changes $Q$
by pullback.

Thus we obtain the classical problem of normal forms of functions with
respect to area preserving transformations on the plane, locally near the origin $0$.
If $Q=Q(t,u_4)$ is non-degenerate at $0$, it can be transformed to the function $Q=t$.
If the origin is a nondegenerate stable equilibrium, the Birkhoff normal form
(in the simplest case of 1 degree of freedom resonances are absent)
allows to choose canonical coordinates $(t,u_4)$ with $Q=Q(\rho)$, where $\rho^2=t^2+u_4^2$;
for instance we get the family $Q=\rho^m$.
In the case of degenerate isolated critical point or non-isolated critical point the normal forms
are too complicated (there are uncountably many non-equivalent such); for instance
we get the family $Q=t^m$. In the ultimate
degenerate case $Q=\op{const}$.
Thus even the local problem yields more than continuum of non-equivalent integrable equations
(deformations of the general heavenly equation)
 $$
(u_{12}u_{34}-u_{13}u_{24})/(u_{12}u_{34}-u_{14}u_{23})=Q(t,u_4).
 $$
In the global case, the problem is reduced to  classification of Hamiltonian systems with
1 degree of freedom. In the case the Hamiltonian is a Morse function, the problem was solved in
\cite{Kr}. In the general case the solution is unknown.

 \section{Geometry and Integrability}\label{S4}

  \subsection{Integrability via geometry of the solution spaces}\label{SIGS}

Integrability via a Lax pair is closely related to integrability by the method of hydrodynamic reductions.

In \cite{FK} the following conjecture for dispersionless systems was formulated and motivated:

 \smallskip

  {\bf Criterion.}
{\it A 2nd order PDE $F=0$ with 4 independent variables on one scalar function
$u=u(x^1,...,x^4)$, with nondegenerate symbol, and of dispersionless type
is integrable by the method of hydrodynamic reductions iff the conformal structure
associated with the metric $g_{ij}dx^idx^j$, where the matrix $(g_{ij})$ is inverse to
 $(\p^2_{i,j}F)$, is (anti-)self-dual (depending on the choice of orientation)
 on every solution of the system $F=0$. }

 \smallskip

Though not proved in full generality, the above criterion can be used as a test for finding
integrable equations, as the self-duality property is contact invariant and straightforward to
check.

 \begin{theorem}
The symmetric deformations considered in Section \ref{SD} are integrable
by the Criterion iff they are integrable via a Lax pair.
 \end{theorem}

  \begin{proof}
This follows by direct computations showing that on the function $Q$ from the proof
of Theorem \ref{Thm3} integrability condition imposes the same constraints as
the condition of existence of a Lax pair. Let us notice here
that these Lax pairs are of the form motivated by the twistor theory, namely the
2-planes $\langle V,W\rangle$ are totally null with respect to the metric $g$ as in the Criterion
(the space of such null 2-planes at every point is 1-dimensional and is parametrized by the
''spectral parameter'' $\l$). Consequently the integral surfaces form a 3-parameter family of null surfaces,
and this characterizes self-duality.
 % Conversely, the 3-parametric family of the planes
 % uniquely recovers the conformal structure $[g]$ and the Lax pair
  \end{proof}

  \begin{rk}
The criterion can be also used to verify non-integrability of certain PDE.
Consider, for instance, Khokhlov-Zabolotskaya equation
 $$
u_{14}-u\,u_{11}-u_1^2-u_{22}+u_{33}=0.
 $$
This equation is the 4D analog of the 3D integrable equation, known as dispersionless
Kadomtsev-Petviashvilli.
 %(dKP; in some sources it is also called 3D Khokhlov-Zabolotskaya equation).
Notice that the 3D equation is integrable by all known methods (Lax pair, hydrodynamic
integrability, Einstein-Weyl geometry on any solution etc).

It is known that the Khokhlov-Zabolotskaya equation does not pass the hydrodynamic integrability test
(E.Ferapontov: private communication). We can aply the above Criterion,
and it also fails on this equation, signifying its non-integrability.

Notice that the infinite part of the contact symmetry pseudogroup of this equation
is parametrized by 3 functions of 1 argument \cite{Ly}. This is in contrast with the
Monge-Amp\`ere equations considered in this paper, for which integrability
was manifested by the symmetry pseudogroup with the infinite part parametrized by
functions of 2 arguments.
  \end{rk}

 \subsection{Petrov types}

Let us notice that the symmetry algebras of Monge-Amp\`ere equations
considered in Section \ref{SAE} are equally big (for all cases except (O) they are
parametrized by 4 copies of $\op{SDiff}(2)\simeq\Pp/\R$), when considered in full generality.
When restricted to subalgebras of $\op{Sp}(8)$ (which is a natural equivalence group
for Hirota type equations) they are Lie groups of different dimensions:
it was computed in \cite{DF} that
 \begin{gather*}
\dim(\mathfrak{sym}({\rm O}))=16,\ \dim(\mathfrak{sym}({\rm D}))=\dim(\mathfrak{sym}({\rm III}))=13,\\
\dim(\mathfrak{sym}({\rm N}))=14,\ \ \dim(\mathfrak{sym}({\rm II}))=\dim(\mathfrak{sym}({\rm I}))=12, \ {}
 \end{gather*}
where $\mathfrak{sym}(\E)=\mathfrak{cont}(\E)\cap\op{Sp}(8)$ denotes the linear symmetries of a PDE $\E$
(it was noted in \cite{DF} that the six type of Monge-Amp\`ere equations are not
$Sp(8)$-equivalent; basing on the computation in Theorem \ref{Thm1} we conclude that
they are not even contact equivalent; however they can be equivalent by a B\"acklund-type
transformation, and indeed such exists between the first and the second Pleba\'{n}ski equations \cite{Pb}).

The above decrease of dimensions is in accordance with the Penrose incidence diagram
of Petrov types:

 
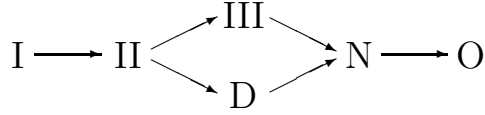
\begin{center}
 \begin{picture}(170,50)
\put(0,20){\large I}
\put(9,25){\vector(1,0){25}}
\put(39,20){\large II}
\put(53,27){\vector(2,1){25}}
\put(53,23){\vector(2,-1){25}}
\put(80,35){\large III}
\put(80,5){\large \,D}
\put(98,39){\vector(2,-1){25}}
\put(98,11){\vector(2,1){25}}
\put(124,20){\large \,N}
\put(140,25){\vector(1,0){25}}
\put(166,20){\large \,O}
 \end{picture}
  \end{center}

The notation for the 6 types of Monge-Amp\`ere equations follows the well-known
stratification of $3\times3$ complex matrices considered modulo similarity
(equivalently: classification of pencils of conics). Four points of intersections of two
distinct conics (counted with multiplicity) form six possible configurations:
the generic case I, where the four points are distinct;
the four degenerations II, III, D, N, when two, three, two disjoint pairs, or
four points come together;
and the trivial case O, when the pencil of conics is constant.

Interpreting the summands $\Pp$ in the symmetry algebra $\g$ of the Monge-Amp\`ere equation
as a point and depth of grading as multiplicity, we obtain from Theorem \ref{Thm1} the
interpretation for the names with the same degeneration scheme: the type $N$ corresponds to
one grading of depth 4 (grades from 0 to 3), \dots, type $I$ corresponds to 4 commuting
(non-graded) pieces of $\Pp$.

 \begin{rk}
The classification of matrices discussed above is applied to classify Weyl tensors of 4-dimensional
conformal structures in Lorenzian signature (Petrov types). We naturally associate conformal
structures to any solution of non-degenerate PDE in 4D (see Criterion in Section \ref{SIGS}).
However the corresponding Petrov type is not fixed as solutions of the same equation vary.
For the second Pleba\'{n}ski equation this can be seen from computations of the Weyl tensor in \cite{Dun}.
 \end{rk}

Relation of gauge pseudogroups to Petrov types was already noticed in \cite{C}, but
our interpretation is apparently more related to the study of singular varieties in \cite{DF}.
From the proof of Theorem\ref{Thm1} we notice that the
Hamiltonians parametrizing the 4 copies of $\op{SDiff}(2)$
are partitioned into sets of $k_1,\dots,k_r$ functions, $k_1+\dots+k_r=4$,
such that any set of $k_i$ functions are defined on the same surface $M^2_i$
(giving the Poisson algebra $\Pp=\op{SDiff}(M_i^2)+\R$).

The origin of these surfaces comes from the geometry described by
Doubrov and Ferapontov. Integrable Monge-Amp\`ere equation in their
interpretation corresponds to a hyperplane tangential to the Pl\"ucker embedding
of the Lagrangian Grassmannian $\Lambda^{10}$ along a subvariety $X^4$ which meets
all $\Lambda^6 \subset\Lambda^{10}$ corresponding to
congruence of Lagrangian subspaces in a symplectic space $V^8$.

Now there is a symplectic splitting $V^8=\oplus V^{2k_i}_i$ and the singular variety $X^4$
consists of all Lagrangian subspaces which intersect $V_i^{2k_i}$ by a Lagrangian subspace
(of dimension $k_i$; in fact there are some more conditions: if $k_i>1$, then
intersection of $X^4$ with $V^{2k_i}_i$ shall meet a fixed isotropic 2-plane $\ell^2$ by a line;
if $k_i>2$, then it shall meet a fixed co-isotropic 4-plane $L^4\supset\ell^2$ by a 3-plane;
the type O corresponds to non-split $V^8$, like type N, but goes a bit different pattern, see \cite{DF}).

The infinite symmetry block $\Pp_i$ consists of Hamiltonians on $V^{2k_i}_i$ and if $k_i>1$
this higher-dimensional symplectic space is constructed from the 2-dimensional space $M^2_i$ by
a natural (though complicated) procedure, see \cite{KM}, which in particular implies that
these parts of symmetries are point transformations; on the contrary for $k_i=1$ the
contribution is general contact, see formulae in the proof of Theorem \ref{Thm1}.

Finally the partitions of 4 into sum of $k_i$ encode the Petrov types:
 \begin{gather*}
(1,1,1,1)\ \thicksim\text{ type I },\quad
(1,1,2)\ \thicksim\text{ type II},\\
(1,3)\ \thicksim\text{ type III},\quad
(2,2)\ \thicksim\text{ type D},\quad
(4)\ \thicksim\text{ type N}
 \end{gather*}
(type O also correspond to $(4)$, but has some degeneracies leading to a bigger pseudogroup
of symmetries).

This finishes explanation of the numerology of equations by Petrov types. The incidence diagram
reflects the degenerations between the equations, e.g. all other integrable Monge-Amp\`ere equations
of Hirota type can be obtained from the general heavenly equation I by some limiting process
and this can be seen from their symmetry algebras.

  \subsection{Recursion operators}\label{rec}

Let us indicate recursion operators for those symmetric deformation that have
hyper-Hermitian property (Theorem~\ref{Thm3}).
In accordance with \cite{KrasilshchikKersten1995}, recursion operators are B\"acklund auto-trans\-formations
for the tangent covering of the equation under study.
To construct recursion operators we use the technique of
\cite{Sergyeyev2015}, see also \cite{MalykhNutkuSheftel2004,MarvanSergyeyev2012,MorozovSergyeyev2014}.
We provide details of computations for the deformation
(\ref{dghe}) of the general heavenly equation.
Its linearization is the restriction of equation
 \begin{gather}
Q\,(u_{23} D^2_{x,t}\varphi + u_{14} D^2_{y,z}\varphi)\notag
-(Q-1)\,(u_{34} D^2_{x,y} \varphi+ u_{12} D^2_{z,t}\varphi)\quad\\
{}-u_{24} D^2_{x,z}\varphi - u_{13} D^2_{y,t}\varphi
-Q_{u_4} (u_{12}\,u_{34} - u_{14}\,u_{23})D_t \varphi=0
\label{linearized_dghe}
 \end{gather}
on (\ref{dghe}). A function $\varphi=\varphi(t,x,y,z,u,u_1,u_2,u_2,u_4)$ subject to (\ref{linearized_dghe})
is called a {\it generator of infinitesimal contact symmetry} of equation (\ref{dghe}).

The covering with `spectral parameter` $\lambda = \mathrm{const}$, $\lambda \not \in \{0,-1\}$
corresponding to the vector fields (\ref{dghe_Lax_pair}) is defined by the system
 \begin{equation}
\left\{\begin{array}{lcl}
r_x &=& \displaystyle{
\frac{Q}{\lambda+1}\,\frac{u_{14}}{u_{24}}\,r_y+\Bigl(1-\frac{Q}{\lambda+1}\Bigr)\,\frac{u_{12}}{u_{24}}\,r_t,
\phantom{\frac{\frac{1}{1}}{\frac{1}{1}}}} \\
r_z &=& \displaystyle{
\frac{Q-1}{\lambda}\,\frac{u_{34}}{u_{24}}\,r_y+\Bigl(1-\frac{Q-1}{\lambda}\Bigr)\,\frac{u_{23}}{u_{24}}\,r_t,
\phantom{\frac{\frac{1}{1}}{\frac{1}{1}}}}
\end{array}\right.
\label{dghe_covering}
 \end{equation}
where $r$ is the coordinate on the fibre of the covering.
A function $\varphi=\varphi(t,x,y,z,u,u_t,u_x,u_y,u_z,r,r_y,r_z,r_{yy},r_{yz},r_{zz},....)$ subject to (\ref{linearized_dghe}) and (\ref{dghe_covering})  is called \cite{KrasilshchikVinogradov1989,Bocharov1999}
a {\it shadow of nonlocal symmetry} of equation (\ref{dghe}).

Let us make the following observation: the function $\varphi=r$ is a shadow of equation (\ref{dghe})
in the covering (\ref{dghe_covering}).
In other words, every solution $r$ to (\ref{dghe_covering}) is a solution to (\ref{linearized_dghe}) as well.

Suppose a solution $r$ to (\ref{dghe_covering}) is a sum of a Laurent series in $\lambda$:
 \begin{equation}
r = \sum \limits_{n=-\infty}^{\infty} \lambda^n \, r_n.
\label{q_expansion}
 \end{equation}
Since (\ref{linearized_dghe}) is independent on $\lambda$, each $r_n$ is a solution to (\ref{linearized_dghe}).  Substituting for (\ref{q_expansion}) into
(\ref{dghe_covering}) yields
\[
\left\{\begin{array}{lcl}
r_{n,x} &=& \displaystyle{
Q\,\frac{u_{14}}{u_{24}}\,r_{n+1,y}-r_{n+1,x}
-(Q-1)\,\frac{u_{12}}{u_{24}}\,r_{n+1,t}+\frac{u_{12}}{u_{24}}\,r_{n,t},
}\\
r_{n,z} &=& \displaystyle{
(Q-1)\,\frac{u_{34}}{u_{24}}\,r_{n+1,y}-(Q-1)\,\frac{u_{23}}{u_{24}}\,r_{n+1,t}
+\frac{u_{23}}{u_{24}}\,r_{n,t}
}
\phantom{\frac{\frac{1}{1}}{\frac{1}{1}}}
\end{array}
\right.
\]
for each $n$. Now fix $n$ and denote $r_n = \psi$ and $r_{n+1} = \varphi$.
The compatibility conditions of the resulting over-determined system
 \begin{equation}\label{Rho_dghe}
\left\{\begin{array}{lcl}
\psi_1 &=& \displaystyle{Q\,\frac{u_{14}}{u_{24}}\,\varphi_2-\varphi_1
-(Q-1)\,\frac{u_{12}}{u_{24}}\,\varphi_4+\frac{u_{12}}{u_{24}}\,\psi_4,}\\
\psi_3 &=& \displaystyle{(Q-1)\,\frac{u_{34}}{u_{24}}\,\varphi_2-(Q-1)\,\frac{u_{23}}{u_{24}}\,\varphi_4
+\frac{u_{23}}{u_{24}}\,\psi_4}
\phantom{\frac{\frac{1}{1}}{\frac{1}{1}}}
\end{array}\right.
 \end{equation}
%  $$
% \left\{\begin{array}{lcl}
% \psi_x &=& \displaystyle{Q\,\frac{u_{14}}{u_{24}}\,\varphi_y-\varphi_x
% -(Q-1)\,\frac{u_{12}}{u_{24}}\,\varphi_t+\frac{u_{12}}{u_{24}}\,\psi_t,}\\
% \psi_z &=& \displaystyle{(Q-1)\,\frac{u_{34}}{u_{24}}\,\varphi_y-(Q-1)\,\frac{u_{23}}{u_{24}}\,\varphi_t
% +\frac{u_{23}}{u_{24}}\,\psi_t}
% \phantom{\frac{\frac{1}{1}}{\frac{1}{1}}}}
% \end{array}\right.
%  $$
coincide with equations (\ref{dghe}) and (\ref{linearized_dghe}). Since $\varphi$ and $\psi$ are solutions to (\ref{linearized_dghe}), we conclude that system (\ref{Rho_dghe}) defines a recursion operator
for symmetries of our deformed equation (\ref{dghe}).

The same technique gives recursion operators for symmetries of the other equations from Theorem \ref{Thm3}.
Namely we obtain the following.

Equation (\ref{dfhe}) has a recursion operator defined by the system
 \[
\left\{\begin{array}{lcl}
\psi_1 &=&\displaystyle{\frac1{Q}\,\Bigl(
u_{13}\,(q\,\psi_4+\varphi_4)-u_{14}\,(q\,\psi_3+\varphi_3)\Bigr),
\phantom{\frac{\frac11}{\frac11}}}\\
\psi_2 &=& \displaystyle{\frac1{Q}\,\Bigl(
u_{23}\,(q\,\psi_4 + \varphi_4) - u_{24}\,(q\,\psi_3+\varphi_3)\Bigr).
\phantom{\frac{\frac11}{\frac11}}}
\end{array}\right.
 \]

Equation (\ref{dmhe}) has a recursion operator defined by the system
 \[
\left\{\begin{array}{lcl}
\psi_3 &=&
u_{33}\,\varphi_1 - \varphi_2 - u_{13}\,\varphi_3,
\phantom{\frac11}\\
\psi_4 &=&
u_{34} \varphi_1 - u_{14}\,\varphi_3 - Q\,\varphi_4.
\phantom{\frac{1^2}{1^2}}
\end{array}\right.
 \]

Equation (\ref{dhhe}) has a recursion operator defined by the system
 \[
\left\{\begin{array}{lcl}
\psi_2 &=&\displaystyle{\frac1{Q\,u_{14}}\,\Bigl(
Q\,u_{12}\,\psi_4+u_{24}\,\varphi_1-u_{12}\,\varphi_4\Bigr),
\phantom{\frac{\frac11}{\frac11}}}
\\
\psi_3 &=& \displaystyle{\frac1{Q\,u_{14}}\,\Bigl(
Q\,u_{13}\,\psi_4 + u_{34}\,\varphi_1-u_{13}\,\varphi_4\Bigr).
\phantom{\frac{\frac11}{\frac11}}}
\end{array}\right.
 \]

 \section{Conclusion}

The abundance of symmetries for dispersionless integrable PDEs considered in this paper is not
the general rule for such equations, but it can be expected for some of them from their geometric nature:
all such integrable equations are reductions of the self-duality equations in 4D \cite{Penrose,MW},
which have lots of symmetry.

Not much is known however on relation of large symmetry and integrability. It is known by folklore
that sufficiently large symmetry pseudogroup (in our case 2 functions of 3 arguments) implies linearizability
of the equation (see \cite[Th. 6.46]{Olver} for one possible formulation).
We expect that existence of several functions of 2 arguments in the symmetry pseudogroup must bear an implication
for integrability. We have not observed more than 4 copies of the algebra $\op{SDiff}(2)$
in the symmetry of non-degenerate 2nd order scalar PDEs in 4D, and we conjecture that 4 is the maximal such value,
which is attained only for integrable equations.

Other types of integrable equations in 4D can be symmetrically deformed in accordance with our method.
We will present more examples in a forthcoming paper.

\appendix
 \section{Multi-component extensions of integrable PDE}

Let us sketch a method for finding multi-component integrable systems, for which ours are reductions.
The idea is simple: take the Lax pair of the given PDE and change the existing functional coefficients to
new all independent variables. Write the commutation relation and obtain the system which possesses Lax pair
by the construction.

We also show other multi-components versions of the heavenly equations and relate this to self-duality.
To keep exposition short we only consider the cases of the first and second Pleba\'nski equations.

 \subsection{Second heavenly equation}\label{A1}
We start with the second Pleba\'nski equation and its Lax pair
 \begin{gather}
u_{ty}+u_{xz}+u_{xx}u_{yy}-u_{xy}^2=0 \label{PbII}\\
T=\p_t+(\l-u_{xy})\,\p_x+u_{xx}\p_y,\quad Z=\p_z+u_{yy}\p_x-(\l+u_{xy})\,\p_y.\notag
 \end{gather}
Change the latter to a more general pair of vector fields
 \begin{equation*}
T=\p_t+(\l-a)\,\p_x+b\,\p_y,\quad Z=\p_z+c\,\p_x-(\l+d)\,\p_y,
 \end{equation*}
where $(a,b,c,d)$ are some functions of $(x,y,z,t)$. The commutativity constraint gives a determined system
of 4 PDE on 4 unknowns. Potentiating two simplest equations in the list we obtain the following integrable
system of 2 second order PDE on 2 unknowns $(v,w)$:
 \begin{equation}\label{lll}
 \begin{array}{c}
v_{ty}+v_{xz}+v_yv_{xx}-(v_x+w_y)\,v_{xy}+w_xv_{yy}=0\vphantom{\frac{A^A}{A^A_A}}\\
w_{ty}+w_{xz}+v_yw_{xx}-(v_x+w_y)\,w_{xy}+w_xv_{yy}=0\vphantom{\frac{A^A_A}{A^A}}
 \end{array}
 \end{equation}
with the Lax pair
 \begin{equation}\label{LPforIIPb}
T=\p_t+(\l-v_x)\,\p_x+w_x\p_y,\quad Z=\p_z+v_y\p_x-(\l+w_y)\,\p_y.
 \end{equation}
The reduction to the original equation is $v=u_y$, $w=u_x$.
The two-component (\ref{lll}) system was obtained from different considerations in \cite{FP,GS}.

The characteristic polynomial of system (\ref{lll}) is the square of the quadric
corresponding to the bi-vector $\p_t\p_y+\p_x\p_z+v_y\p_x^2-(v_x+w_y)\p_x\p_y+w_x\p_y^2$
upon substitution $\p_i\mapsto p_i$. % $(p_tp_y+p_xp_z+v_yp_x^2-(v_x+w_y)p_xp_y+w_xp_y^2)^2$.
According to the method described in Section \ref{SIGS} the inverse tensor is the conformal metric
 \begin{equation}\label{metricIIPb}
g=dt\,dy+dx\,dz-w_xdt^2+(v_x+w_y)dtdz-v_ydz^2
 \end{equation}
naturally associated to the system. This metric is self-dual on every solution of the system on (\ref{lll}),
but it is not Ricci flat. Ricci flatness is equivalent to $v=u_y,w=u_x$, i.e. the second
Pleba\'nski equation (\ref{PbII}).

\smallskip

Consider now the covering corresponding to the Lax pair (\ref{LPforIIPb}):
 \begin{equation} \label{covering_1}
\left\{
\begin{array}{lcl}
q_t  &=&(v_x-\lambda)\,q_x-w_x\,q_y, \\
q_z &=& -v_y\,q_x+(w_y+\lambda)\,q_y
\end{array}
\right.
 \end{equation}
The scaling symmetry of system (\ref{lll}) is not liftable to a contact symmetry of the joint system
(\ref{lll})+(\ref{covering_1}), hence we get
 \begin{lem}
$\lambda$ is a non-removable parameter of the covering (\ref{covering_1}).
 \end{lem}

A different from (\ref{lll}) two-component generalization of the second heavenly equation was proposed by Dunajski in
\cite{DunP}. It is possible to apply Dunajski's construction to the system $T(\lambda)=Z(\lambda)=0$
(by twistor construction the Lax vector fields are tangent to the surfaces $\lambda=\lambda(x,y,z,t)$ in the
extended configuration space with coordinates $(x,y,z,t,\lambda)$,
and so $L_{\mu_1T+\mu_2Z}(\lambda-\lambda(x,y,z,t))=0$ $\forall \mu_1,\mu_2\in\R$)%
\footnote{This system can be seen as Pavlov's eversion \cite[\S2]{PCC} of the covering (\ref{covering_1}),
which via the change $\lambda\to q$ transforms linear covering to a nonlinear Zakharov type Lax pair.} and we get:
 \begin{equation}\label{covering_2}
\left\{
\begin{array}{lcl}
q_t  &=&(v_x-q)\,q_x-w_x\,q_y+u_x,\\
q_z &=& -v_y\,q_x+(w_y+q)\,q_y-u_y.
\end{array}
\right.
 \end{equation}
(this function $u$ has no relation to the function $u$ in equation (\ref{PbII}))

 \begin{prop}\label{P2}
System (\ref{covering_2}) defines a covering for the following three-component generalization of
the second heavenly equation
 \begin{equation*}
\left\{\begin{array}{l}
u_{xz}+u_{ty}+v_y\,u_{xx}-(v_x+w_y)\,u_{xy}+w_x\,u_{yy}=0,\\
v_{xz}+v_{ty}+v_y\,v_{xx}-(v_x+w_y)\,v_{xy}+w_x\,v_{yy}+u_y=0,\\
w_{xz}+w_{ty}+v_y\,w_{xx}-(v_x+w_y)\,w_{xy}+w_x\,w_{yy}+u_x=0.
\end{array}\right.
 \end{equation*}
 \end{prop}

Notice that all three equations have the same symbol bi-vector, which is dual to the metric (\ref{metricIIPb}).
We will comment on its properties below.

 \subsection{First heavenly equation}\label{A2}
% Similarly we can consider extensions of the other four equations.
For the first Pleba\'nski equation (D) the Lax pair is a particular case of a pair of vector fields
 $$
X=\l\p_x+a\p_t+b\p_z,\quad Y=\l\p_y+c\p_t+d\p_z.
 $$
The commutation condition writes as 4 equations, 2 of which can be resolved $a=v_x$, $b=w_x$, $c=v_y$, $d=w_y$,
whence we get the following integrable system together with its Lax pair:
 \begin{gather*}
v_xv_{ty}-v_yv_{tx}+w_xv_{yz}-w_yv_{xz}=0=v_xw_{ty}-v_yw_{tx}+w_xw_{yz}-w_yw_{xz},\\
X=\l\p_x+v_x\p_t+w_x\p_z,\quad Y=\l\p_y+v_y\p_t+w_y\p_z.
 \end{gather*}
The corresponding conformal structure
 \begin{equation}\label{metricIPb}
g=w_xdt\,dx+w_ydt\,dy-u_xdx\,dz-u_ydy\,dz
 \end{equation}
is self-dual on every solution of the system.

The above two-component generalization is equivalent to the considered in \cite{GS}
(this paper also has a three-component generalization different from the one below;
other two-component generalizations of the first heavenly equation were studied in \cite{Park1991}).
Its Lax pair can be written as %the covering
 \begin{equation}
\left\{\begin{array}{lcl}
q_x  &=&  -\lambda\,(v_x\,q_t+w_x\,q_z), \\
q_y &=&  -\lambda\,(v_y\,q_t+w_y\,q_z).
\end{array}\right.
\label{covering_3}
 \end{equation}

 \begin{lem}
$\lambda$ is a non-removable parameter of the covering (\ref{covering_3}).
 \end{lem}

% Again, it is possible to apply Dunajski's construction to Pavlov's eversion to (\ref{covering_3})
% resulting in the system
Applying the same trick as in \ref{A1} to (\ref{covering_3}) we obtain the system
 \begin{equation}\label{covering_4}
\left\{\begin{array}{l}
q_x+q\,(v_x\,q_t+w_x\,q_z+q\,u_x)=0,\\
q_y+q\,(v_y\,q_t+w_y\,q_z+q\,u_y)=0.
\end{array}\right.
 \end{equation}

 \begin{prop}\label{P3}
System (\ref{covering_4}) defines a covering for the following three-component generalization of
the first heavenly equation
 \begin{equation*}
\left\{\begin{array}{l}
w_x\,u_{yz}-v_y\,u_{tx}+v_x\,u_{ty}-w_y\,u_{xz}=0,\\
w_x\,v_{yz}-v_y\,v_{tx}+v_x\,v_{ty}-w_y\,v_{xz}-v_y\,u_x+v_x\,u_y=0,\\
w_x\,w_{yz}-v_y\,w_{tx}+v_x\,w_{ty}-w_y\,w_{xz}-w_y\,u_x+w_x\,u_y=0.
\end{array}\right.
 \end{equation*}
 \end{prop}
The three equations have the same symbol bi-vector, which is dual to the metric (\ref{metricIPb}).
They also yield self-dual metrics.

\smallskip

Alternatively, we can start from the following Lax pair in the form of vector fields,
in which $(a,b,c,d,q,r)$ are arbitrary functions of $(x,y,z,t)$:
 \begin{equation}\label{genLP}
X=\p_x+(\l+q)(a\p_t+b\p_z),\quad Y=\p_y+(\l+r)(c\p_t+d\p_z)
 \end{equation}
that generalizes the Lax pair (\ref{ourD}) for our deformation (D). Commutativity condition $[X,Y]=0$
of (\ref{genLP}) writes as the following determined system of 6 PDE on 6 unknowns:
 \begin{gather*}
a\,c_t+b\,c_z-c\,a_t-d\,a_z=0,\quad a\,d_t+b\,d_z-c\,b_t-d\,b_z=0,\\
aq\,(dr)_t-cr\,(bq)_t+bq\,(dr)_z-dr\,(bq)_z+(dr)_x-(bq)_y=0,\\
aq\,(cr)_t-cr\,(aq)_t+bq\,(cr)_z-dr\,(aq)_z+(cr)_x-(aq)_y=0,\\
ac(r-q)_t+(q+r)(ac_t+bc_z-ca_t-da_z)-adq_z+bcr_z-a_y+c_x=0,\\
adr_t-bcq_t+(q+r)(ad_t-cb_t+bd_z-db_z)+bd(r-q)_z-b_y+d_x=0.
 \end{gather*}
Again if we compute the characteristic polynomial of this system, we recover a conformal structure.
Indeed, it the polynomial (of degree 6) contains two linear factors and the square of the quadric
corresponding to the following bi-vector upon substitution $\p_i\mapsto p_i$:
 $$
(q-r)\bigl(ac\,\p_t^2+(ad-bc)\p_t\p_z+bd\,\p_z^2\bigr)-a\,\p_t\p_y-b\,\p_y\p_z+c\,\p_t\p_x+d\,\p_x\p_z.
 $$
The inverse tensor is the conformal structure naturally associated with the above $6\times6$ system:
 \begin{equation}\label{SDsys}
g=b\,dt\,dx+d\,dt\,dy-(q-r)(ad-bc)\,dx\,dy-a\,dx\,dz-c\,dy\,dz.
 \end{equation}

 \begin{theorem}
Metric (\ref{SDsys}) is self-dual on every solution of the above $6\times6$ system.
Likewise metric (\ref{metricIIPb}) is self-dual on every solution of the system of Proposition \ref{P2},
and metric (\ref{metricIPb}) is self-dual on every solution of system of Proposition \ref{P3}.
 \end{theorem}
This results from a straightforward computation.

 \begin{rk}
The general solution of the above system depends on 6 arbitrary functions of 3 arguments, and therefore we
obtain self-dual metrics parametrized by that many functions. But this is precisely the functional freedom
in the general self-dual metric \cite{DFK}. Thus we expect that solutions of our $6\times6$ system
parameterize an open stratum of the space of self-dual metrics in 4D.

Similarly we expect that three component generalizations of the first and second heavenly equations,
which appear in Propositions \ref{P2} and \ref{P3}, give an alternative form of the metric self-duality equations.
 \end{rk}

Other equations {\rm III}, {\rm II}, {\rm I} admit similar generalizations, and we will not specify them.
The construction provides integrable systems and various parameterizations of self-dual metrics.
In the main body of the paper we have been interested in the second order integrable PDE leading
to self-dual conformal structures, which are deformations of the five "heavenly" type equations.

From the viewpoint presented in this appendix these are reductions of multi-component/higher-order systems,
yet to find such reductions is non-trivial (we can do this a-posteriori). Most important feature of (non-degenerate)
second order equations is that they have a canonical conformal structure associated with every solution.
This is not so for multi-component/higher order systems, though we organized our generalizations (extensions)
in a proper way so that such a conformal structure does exist.

\vskip0.3cm

%%%%%%%%%%%%%%%%%%%%%%%%%%%%%%%%%%%%%%%%%%%%%%%%%%%%%%%%%%%%%%%%%%%%%%%%%%%%

% \vspace{-3pt} \hspace{-20pt} {\hbox to 12cm{ \hrulefill }}
%\vspace{8pt}

%{\footnotesize \hspace{-15pt} Boris Kruglikov, Oleg Morozov:
%Institute of Mathematics and Statistics, University of Troms\o, Troms\o\ 90-37, Norway.

%\hspace{-10pt} E-mails: \ \textsc{boris.kruglikov@uit.no}, \quad
%\textsc{oleg.morozov@uit.no}.

\end{document}